\documentclass[11pt,fullpage]{article}
\usepackage{array}
\usepackage{latexsym}
\usepackage{amsfonts}
\usepackage{amsmath}
\usepackage{amssymb}
\usepackage{amsthm}

\usepackage{geometry}
\usepackage{algorithm}
\usepackage{algorithmic}
\usepackage{color}
\usepackage{mathrsfs}

\theoremstyle{plain}
\newtheorem{theorem}{Theorem}[section]

\newtheorem{lemma}[theorem]{Lemma}

\newtheorem*{theorem*}{Theorem}

\theoremstyle{definition}
\newtheorem*{definition*}{Definition}
\newtheorem{definition}[theorem]{Definition}

\newtheorem{example}[theorem]{Example}
\newtheorem{remark}[theorem]{Remark}

\theoremstyle{remark}
\newtheorem*{remark*}{Remark}
\newtheorem*{property*}{Property}

\newcommand{\R} {\ensuremath{\mathbb{R}}} 

\newcommand{\eps}{\ensuremath{\varepsilon}}

\newcommand{\cel}[1]{{\lceil {#1} \rceil}}

\newcommand{\refeq}[1]{(\ref{#1})}

\newcommand{\LDOTS}{,\,\ldots,\,}     

\newcommand{\X}{{\mathbf{X}}}
\newcommand{\Y}{{\mathbf{Y}}}
\newcommand{\M}{{\mathcal{M}}}
\newcommand{\C}{{\mathcal{C}}}
\newcommand{\distrib}[1]{{\mathbf{\Delta}(#1)}}

\newcommand{\puto}{{\eps^*}}

\title{Privacy-Compatibility For General Utility Metrics}
\author{
Robert Kleinberg \footnotemark[1] \footnotemark[2]  \and
Katrina Ligett \footnotemark[1] \footnotemark[3]
}
\footnotetext[1]{Department of Computer Science, Cornell University, Ithaca NY 14853.  E-mail: \texttt{\{rdk,katrina\}@cs.cornell.edu}.}
\footnotetext[2]{Supported by
NSF Award CCF-0643934, an Alfred P. Sloan Foundation Fellowship, a Microsoft Research New Faculty Fellowship, and a grant from the Air Force Office of Scientific Research.}
\footnotetext[3]{Supported by an NSF Computing Innovation Fellowship 
(NSF Award 0937060)
and an NSF Mathematical Sciences Postdoctoral Fellowship (NSF Award 1004416).}

\begin{document}
\date{}
\maketitle
\thispagestyle{empty}

\begin{abstract}
In this note, we present a complete characterization of the utility metrics that allow for non-trivial differential privacy guarantees.
\end{abstract}

\newpage
\setcounter{page}{1}

\section{Introduction}
\label{sec:intro}

The field of data privacy is, at its heart, the study of tradeoffs
between utility and privacy.  The theoretical computer science
community has embraced a strong and compelling definition of privacy
--- differential privacy~\cite{DN03,DMNS06} --- but utility
definitions, quite naturally, depend on the application at hand.  For
a given function $f$, can we achieve arbitrarily close to perfect
utility by relaxing the privacy parameter sufficiently?  We show that
this question has a satisfyingly simple answer: yes, if and only if
the image of $f$ has compact completion.  Furthermore, in this case
there exists a single base measure $\mu$ such that conventional
exponential mechanisms based on $\mu$ are capable of achieving
arbitrarily good utility.

\section{Definitions}
\label{sec:defs}

We are given two metric spaces $(\X,\rho)$ and 
$(\Y,\sigma)$ and a continuous function $f: \X \rightarrow \Y.$
We think of the input database as being an element $x \in \X$,
and our goal is to disclose an approximation to the value of 
$f(x)$ while preserving privacy.  
To allow for a cleaner exposition, we will assume
throughout this paper that $f$ has Lipschitz constant $1$,
i.e.\ $\sigma(f(x),f(z)) \leq \rho(x,z)$ for all $x,z \in \X.$
All of our results generalize to arbitrary Lipschitz
continuous functions, an issue that we return to in
Remark~\ref{remark:sensitivity}.

\begin{definition}  \label{def:mech}
A \emph{mechanism} is a function $\M \,:\, \X
\rightarrow \distrib{\Y}$, 
where $\distrib{\Y}$ denotes the set of all
Borel probability measures on $Y$.  
For a point $x \in \X$, we will often
denote the probability measure $\M(x)$ using the alternate notation
$\M_x$.
\end{definition}
 
\begin{definition} \label{def:priv-ut}
For $\eps>0$, we say that a mechanism $\M$
\emph{achieves $\eps$-differential privacy} if the following
relation holds for every $x,z \in X$ and every Borel set $T \subseteq \Y$:
\begin{equation} \label{eq:priv}
\M_x(T) \leq e^{\eps \rho(x,z)} \M_z(T).\footnote{A number of results in the literature, including recent work of Roth
and Roughgarden~\cite{RR10} on mechanisms for predicate queries,
achieve only a weakened definition of privacy known as $(\eps,
\delta)$-differential privacy; such results do not fit in the
framework presented here.}
\end{equation}
For $\gamma,\delta>0$, we say that 
$\M$ \emph{achieves $\gamma$-utility with probability
at least $1-\delta$} if the following relation holds for every
$x \in X$:
\begin{equation} \label{eq:ut}
\M_x(B_{\sigma}(f(x),\gamma)) \geq 1-\delta.
\end{equation}
We abbreviate this relation by saying 
that $\M$ achieves $(\gamma,\delta)$-utility.
\end{definition}

\begin{definition} \label{def:tradeoff}
Given a function $f : \X \rightarrow \Y$,
the \emph{privacy-utility tradeoff} of $f$
is the function 
$$
\puto(\gamma,\delta) = \inf \{ \eps > 0 \,|\, \exists
\mbox{ a mechanism $\M$ satisfying
$\eps$-differential privacy and $(\gamma,\delta)$-utility} \},
$$
where the right side is interpreted as $\infty$ if the set 
in question is empty. 
\end{definition}

\begin{remark} \label{remark:sensitivity}
In prior work on differential privacy, it is more customary
to express differential privacy guarantees in terms of an
\emph{adjacency relation} on inputs, rather than a metric
space on the inputs.  In this framework, the 
sensitivity of $f$ (the maximum of $|f(a)-f(b)|$ over all
adjacent pairs $a,b$) plays a pivotal role in determining
the privacy achieved by a mechanism.  The Lipschitz constant
of $f$ plays the equivalent role in our setting.

One could of course equate the two frameworks by defining
the privacy metric $\rho$ to be the shortest-path metric
in the graph defined by the adjacency relation.  This
would equate the Lipschitz constant of $f$ with its
sensitivity.  However, it is much more convenient to
describe our mechanisms and their analysis under the
assumption that $f$ has Lipschitz constant $1$; for any
Lipschitz continuous $f$ this can trivially be achieved
by rescaling both $\rho$ and the corresponding privacy
bound by $C$, the Lipschitz constant of $f$.

Thus, for example, if one is given a function $f$
and wishes to know whether there exists a mechanism
achieving $\eps$-differential privacy and $(\gamma,\delta)$-utility,
the answer is yes if and only if $\eps/\eps^*(\gamma,\delta)$ is 
greater than the Lipschitz constant (i.e.,~sensitivity) of $f$.
In cases where the sensitivity $\Delta_f$
depends on the number of 
points in an input database, $N$, the
relation $\eps/\eps^*(\gamma/\delta) \geq \Delta_f$ can be
used to solve for $N$ in terms of the parameters
$\eps,\gamma,\delta.$  For example, in many papers
(e.g.~\cite{BLR08}) $\Delta_f = 1/N$ and then we find
that $N = \eps^*(\gamma,\delta)/\eps$ is the minimum
number of points in the input database necessary to
achieve $\eps$-differential privacy and $(\gamma,\delta)$-utility.
\end{remark}

\begin{remark}
Our definition of utility captures many prior formulations.  
For setings where the output space is simply $\mathbb{R}$, the
  traditional utility metric reflecting the difference between the
  given answer and the true answer is easily captured in our
  framework.  A variety of prior work on problems involving more
  complex outputs can also be cast as measuring utility in a metric
  space.  For example, Blum et al.~\cite{BLR08} propose utility with
  respect to a concept class $\mathcal{H}$, and define the utility of
  a candidate output database $y$ on an input $x$ as $\max_{h \in
    \mathcal{H}} | h(x) - h(y)|$.  This setup can be viewed as mapping
  input databases $x$ to vectors $(h_1(x), h_2(x), \ldots)$ and taking the
  utility metric $\sigma$ to be the $L^{\infty}$ metric on output vectors.
  Hardt and Talwar~\cite{HT10} use $L^2$ as their utility metric, but
  whereas they compute the mean square (or $p$-th moment) of its
  distribution, we define disutility to be the probability that the $\sigma$ value
  exceeds $\gamma$.
\end{remark}

\begin{definition}  Given a  measure $\mu$ on $\X$,
and a scalar $\beta>0$, the \emph{(conventional) exponential mechanism} 
$\C^{\mu;\beta}$ is given by the formula:
\begin{equation} \label{eq:exp-mech}
\C_x^{\mu;\beta}(T) = 
\frac{\int_T e^{-\beta \sigma(f(x),y)} \, d \mu(y)}
{\int_{\Y} e^{-\beta \sigma(f(x),y)} \, d \mu(y)},
\end{equation}
provided that the denominator is finite.  Otherwise
$\C_x^{\mu;\beta}$ is undefined.\footnote{We use the word ``conventional'' here to refer to the rich subclass of exponential mechanisms whose score function is $\sigma$; however, not all exponential mechanisms fall in this class.}
\end{definition}
The differential privacy guarantee for exponential mechanisms
is given by the following theorem, whose proof parallels the
original proof of McSherry and Talwar~\cite{McST07} and is
given in the Appendix.
\begin{theorem} \label{thm:exp-mech}
If $f$ has Lipschitz constant $C$ then the 
conventional exponential mechanism
$\C^{\mu;\beta}$ is
$(2 C \beta)$-differentially private for every $\mu$.
\end{theorem}

\section{A topological criterion for privacy-compatibility}
\label{sec:nasc}

A surprising result of Blum et al.~\cite{BLR08} shows that, in the
natural setting of one-dimensional range queries over continuous
domains, \emph{no} mechanism can simultaneously achieve non-trivial
privacy and utility guarantees.  What is it about this application
that makes privacy fundamentally impossible?  In this section, we
introduce a definition of \emph{privacy-compatibility} and give a
complete characterization of the applications that satisfy this
definition.

\begin{definition}
We say that $f$ is 
\emph{privacy-compatible} if $\eps^*(\gamma,\delta) < \infty$
for all $\gamma, \delta > 0$.
\end{definition}

Suppose that $f$ is Lipschitz continuous and that 
the metric space $(\X,\rho)$ is bounded.  We now
prove that $f$ is privacy-compatible if and only
if the completion of the metric space $f(\X)$ is
compact. 
Observe that rescaling the metrics $\rho,\sigma$ does not
affect the question of whether $f$ is privacy-compatible
nor whether $f(\X)$ has compact completion, but it does 
rescale the Lipschitz constant of $f$ and the diameter of $\X$.  
Accordingly,
we may assume without loss of generality that the Lipschitz constant
of $f$ and the diameter of $\X$ are both bounded above by $1$, i.e.
\begin{equation} \label{eq:rhosigma}
\sigma(f(x_1),f(x_2)) \leq \rho(x_1,x_2) \leq 1
\end{equation}
for all $x_1,x_2 \in \X.$ 

\begin{definition}  A probability measure $\mu$ on a metric
space $(\X,\sigma)$ is \emph{uniformly positive} if it is 
the case that for all $r>0$,
\[
\inf_{x \in X} \mu(B_{\sigma}(x,r)) > 0.
\]
\end{definition}

\begin{example}  The uniform measure on $[0,1]$ is uniformly
positive.  The Gaussian measure on $\R$ is not uniformly 
positive because one can find intervals of width $2r$ with
arbitrarily small measure by taking the center of the
interval to be sufficiently far from $0$.
\end{example}




\begin{theorem} 
If the Lipschitz constant of $f$ and the diameter of $X$ are both bounded above by $1$,
then the following are equivalent:
\begin{enumerate}
\item \label{tfae:1}
$f$ is privacy-compatible;
\item \label{tfae:2}
For every $\gamma,\delta>0$, there is a conventional exponential
mechanism that achieves 
$(\gamma,\delta)$-utility; 
\item \label{tfae:3}
There exists a uniformly positive measure on $(f(\X),\sigma)$;
\item \label{tfae:4}
The completion of $(f(\X),\sigma)$ is compact.
\end{enumerate}
\label{thm:qualitative}
\end{theorem}
\begin{proof}
For simplicity, throughout the proof we assume 
without loss of generality that $\Y = f(\X)$.  
The notation $B(y,r)$
denotes the ball of radius $r$ around $y$ in the
metric space $(\Y,\sigma)$.

 {\bf \refeq{tfae:2} $\Rightarrow$ \refeq{tfae:1}} 
The exponential mechanism $\M^{\mu;\beta}$ achieves
$(2 \beta)$-differential privacy.  

 {\bf \refeq{tfae:3} $\Rightarrow$ \refeq{tfae:2}}
For $\mu$ a uniformly positive measure on $(Y,\sigma)$,
and $\gamma,\delta>0$, let $m = \inf_{y \in \Y} 
\mu(B(y, \gamma/2))$ and let
$\beta = \frac{2}{\gamma} \ln \left( \frac{1}{\delta m} \right).$
We claim that the exponential mechanism $\M = \M^{\mu;\beta}$ achieves
$(\gamma,\delta)$-utility.  To see this,
let $x \in \X$ be an arbitrary point, let $z = f(x)$, and let 
\[
a = \int_{B(x,\gamma)} e^{-\beta \sigma(z,y)} \, d \mu(y) \quad \qquad
b = \int_{\X \setminus B(x,\gamma)} e^{-\beta \sigma(z,y)} \, d \mu(y).
\]
We have 
\begin{align*}
a &\geq \int_{B(z,\gamma/2)} e^{- \beta \sigma(z,y)} \, d \mu(y) 
\geq \int_{B(z,\gamma/2)} e^{- \beta \gamma / 2} \, d \mu(y) =
e^{-\beta \gamma/2} \mu(B(z,\gamma/2)) \geq 
e^{-\beta \gamma/2} m \\
b &< \int_{Y} e^{-\beta \gamma} \, d \mu(y) = e^{-\beta \gamma}.
\end{align*}
Hence, for every $x \in \X$,
\begin{align*}
\M_x(B(f(x),\gamma)) = \frac{a}{a+b} &= 1 - \frac{b}{a+b}
> 1 - \frac{e^{-\beta \gamma}}{e^{-\beta \gamma / 2}m} 
= 1 - \frac{1}{e^{\beta \gamma / 2}m}
= 1 - \delta.
\end{align*}

 {\bf \refeq{tfae:4} $\Rightarrow$ \refeq{tfae:3}}
We use the following fact from the topology of metric spaces:
a complete metric space is compact if and only, for every $r$, 
if it has a finite covering by balls of radius $r$.
(See Theorem~\ref{thm:top} in the Appendix.)
For $i=1,2,\ldots,$ let $C_i = \{y_{i,1} \LDOTS y_{i,n(i)}\}$
be a finite set of points such that the balls of radius $2^{-i}$
centered at the points of $C_i$ cover $\Y$.  Now define a 
probability measure $\mu$
supported on the countable set $C = \cup_{i=1}^{\infty} C_i,$ by 
specifying that for $y \in C,$
$
\mu(y) = \sum_{i \,:\, y \in C_i} \left( \frac{1}{2^i n(i)} \right).
$
Equivalently, one can describe $\mu$ by saying that a procedure for
randomly sampling from $\mu$ is to flip a fair coin until heads 
comes up, let $i$ be the number of coin flips, and sample a point 
of $C_i$ uniformly at random.  We claim that $\mu$ is uniformly
positive.  To see this, given any $r>0$ let $i = \cel{\log_2(1/r)},$
so that $2^{-i} \leq r.$  For any point $y \in \Y$, there exists
some $j \; (1 \leq j \leq n(i))$ such that $y \in B(y_{i,j},2^{-i}).$
This implies that $B(y,r)$ contains $y_{i,j}$, hence
$
\mu(B(y,r)) \geq \mu(y_{i,j}) \geq \frac{1}{2^i n(i)}.
$
The right side depends only on $r$ (and not on $y$), hence
$\inf_{y \in \Y} \mu(B(y,r))$ is strictly positive, as desired.

 {\bf \refeq{tfae:1} $\Rightarrow$ \refeq{tfae:4}}
We prove the contrapositive.  Suppose that the completion of
$\Y$ is not compact.  Once again using point-set topology 
(Theorem~\ref{thm:top}) this implies that there exists
an infinite collection of pairwise disjoint balls of radius $r$,
for some $r>0$.  Let $y_1,y_2,\LDOTS$ be the centers of these
balls.  By our assumption that $\Y = f(\X)$, we may choose
points $x_i$ such that $y_i = f(x_i)$ for all
$i \geq 1$.  Suppose we are given a mechanism $\M$ that achieves 
$r$-utility with probability at least $1/2$.  For every 
$\alpha>0$ we must show that $\M$ does not achieve 
$\alpha$-differential privacy.  The relation
$
\sum_{i=1}^{\infty} \M_{x_1}(B(y_i,r)) \leq 1
$
implies that there exists some $i$ such that 
\begin{equation} \label{eq:tfae.1}
\M_{x_1}(B(y_i,r)) < e^{- \alpha}/2.
\end{equation}
The fact that $\M$ achieves $r$-utility with probability at 
least $1/2$ implies that 
\begin{equation} \label{eq:tfae.2}
\M_{x_i}(B(y_i,r)) > 1/2.
\end{equation}
Combining \eqref{eq:tfae.1} with \eqref{eq:tfae.2} leads to
\begin{equation} \label{eq:tfae.3}
\M_{x_i}(B(y_i,r)) > e^{\alpha} \M_{x_1}(B(y_i,r)) \geq
e^{\alpha \rho(x_i,x_1)} \M_{x_1}(B(y_i,r)),
\end{equation}
hence $\M$ violates $\alpha$-differential privacy.
\end{proof}

\bibliographystyle{abbrv}
\bibliography{topobib}

\appendix
\section{Appendix}
\label{sec:topo}

\begin{lemma} \label{lem:em-privacy}
If $f : \X \rightarrow \Y$ has Lipschitz constant $1$, then 
the conventional exponential mechanism $\M^{\mu;\beta}$ achieves
$(2 \beta)$-differential privacy.
\end{lemma}
\begin{proof}
The proof follows
the original proof of McSherry and Talwar~\cite{McST07}.
The triangle inequality implies that for any $x,z$
\begin{align*}
\int_T e^{-\beta \sigma(f(x),y) \, d\mu(y)} & \leq
\int_T e^{-\beta [\sigma(f(z),y) - \sigma(f(x),f(z))] } \, d\mu(y)\\
& =
e^{\beta \sigma(f(x),f(z))} \int_T e^{-\beta \sigma(f(x),y)} \, d\mu(y)\\
&\leq
e^{\beta \rho(x,z)} \int_T e^{-\beta \sigma(f(z),y)} \, d\mu(y) \\
\int_{\Y} e^{-\beta \sigma(f(x),y) \, d\mu(y)} & \geq
\int_{\Y} e^{-\beta [\sigma(f(z),y) + \sigma(f(x),f(z))] } \, d\mu(y)\\
&=
e^{-\beta \sigma(f(x),f(z))} \int_{\Y} e^{-\beta \sigma(f(x),y)} \, d\mu(y)\\
&\geq
e^{-\beta \rho(x,z)} \int_{\Y} e^{-\beta \sigma(f(z),y)} \, d\mu(y).
\end{align*}
The inequality $\M_x(T) \leq e^{2 \beta \rho(x,z)} \M_z(T)$ follows
upon taking the quotient of these two inequalities.
\end{proof}

\begin{theorem} \label{thm:top}
For a metric space $(\X,\sigma)$, the following are equivalent:
\begin{enumerate}
\item \label{top:1}
The completion of $\X$ is a compact topological space.
\item \label{top:2}
For every $r>0$, $\X$ can be covered by a finite collection
of balls of radius $r$.
\item \label{top:3}
For every $r>0$, $\X$ does not contain an infinite collection
of pairwise disjoint balls of radius $r$.
\end{enumerate}
\end{theorem}

\begin{proof}
{\bf \refeq{top:2} $\Rightarrow$ \refeq{top:1}}
Assume that property \refeq{top:2} holds.
Recall that a metric space is compact if and only if 
every infinite sequence of points has a convergent subsequence,
and it is complete if and only if every Cauchy sequence
is convergent.  Thus, we must prove that every infinite
sequence $x_1,x_2,\ldots$ in $\X$ has a Cauchy subsequence.
We can use a pigeonhole-principle argument
to construct the Cauchy subsequence.  In fact,
the construction will yield a sequence of 
points $z_1,z_2,\ldots$ and sets $S_1,S_2,\ldots$
such that the diameter of $S_k$ is at most $1/k$
and $z_i \in S_k$ for all $i \geq k$; these two
properties immediately imply that $z_1,z_2,\ldots$
is a Cauchy sequence as desired.

The construction begins by defining $S_0 = \X$.  Now,
for any $k > 0$, assume inductively that we have a set
$S_{k-1}$ such that the relation $x_i \in S_{k-1}$ is
satisfied by infinitely many $i$.  Let $B_1,B_2,\ldots,B_{n(k)}$ 
be a finite collection of balls of radius $\frac{1}{2k}$
that covers $\X$.  There must be at least one value of 
$j$ such that the relation $x_i \in S_{k-1} \cap B_j$
is satisfied by infinitely many $i$.  Let $S_k = S_{k-1} \cap B_j$
and let $z_k$ be any point in the sequence $x_1,x_2,\ldots$
that belongs to $S_k$ and occurs strictly later in the 
sequence than $z_{k-1}$.  This completes the construction
of the Cauchy subsequence and establishes that the 
completion of $\X$ is compact.

{\bf \refeq{top:1} $\Rightarrow$ \refeq{top:3}}
If $\X$ contains an infinite collection of pairwise disjoint
balls of radius $r$, then the centers of these balls form an
infinite set with no limit point in $\X$, violating compactness.

{\bf \refeq{top:3} $\Rightarrow$ \refeq{top:2}}
Given $r > 0$, let $B(x_1,r/2) \LDOTS B(x_n,r/2)$
be a maximal collection of disjoint balls of radius $r/2$.  
(Such a collection must be finite, by property \refeq{top:3}.)
The balls $B(x_1,r) \LDOTS B(x_n,r)$ cover $\X$, because
if there were a point $y \in \X$ not covered by these balls,
then $B(y,r/2)$ would be disjoint from 
$B(x_i,r/2)$ for $i=1 \LDOTS n,$ contradicting the maximality
of the collection.
\end{proof}

\end{document}